\documentclass{revtex4}
\usepackage{amsmath,amsthm,amssymb}
\usepackage{rotating}
\usepackage{hhline}
\usepackage{graphicx}

\usepackage{amsfonts}
\newtheorem{theorem}{Theorem}[section]
\newtheorem{lemma}[theorem]{Lemma}
\newtheorem{proposition}[theorem]{Proposition}

\theoremstyle{remark}

\theoremstyle{definition}
\newtheorem{definition}[theorem]{Definition}

\theoremstyle{example}
\newtheorem{example}[theorem]{Example}

\theoremstyle{notation}

\newcommand{\bra}[1]{\langle#1|}
\newcommand{\ket}[1]{|#1\rangle}
\begin{document}

\title{ Random walks in finite Abelian groups with Birkhoff subpolytopes of doubly stochastic matrices and their physical implementation}            
\author{A. Vourdas}
\affiliation{Department of Computer Science,\\
University of Bradford, \\
Bradford BD7 1DP, United Kingdom\\a.vourdas@bradford.ac.uk}

\begin{abstract}
Random walks in a finite Abelian group $G$ are studied.
They use Markov chains with doubly stochastic transition matrices, in a Birkhoff subpolytope ${\cal B}(G)$ associated with the group $G$.
It is shown that all future probability vectors belong to a polytope which does not depend on the transition matrices, and which shrinks during time evolution.
Various quantities are used to describe the probability vectors: the majorization preorder, Lorenz values and the Gini index, entropic quantities, and the total variation distance.
The general results are applied to the additive group ${\mathbb Z}(d)$, and to the  Heisenberg-Weyl group $HW(d)/{\mathbb Z}(d)$.
A physical implementation of random walks in ${\mathbb Z}(d)$ that involves a sequence of non-selective projective measurements, is discussed.
A physical implementation of random walks in the Heisenberg-Weyl group $HW(d)/{\mathbb Z}(d)$ using a sequence of non-selective POVM 
measurements with coherent states, is also presented.

\end{abstract}
\maketitle

\section{Introduction}

Random walks is a very wide subject and has been studied in various contexts.
They have applications in physics (e.g., Brownian motion), Computer Science (e.g., search algorithms), finance (e.g., stock market), meteorology, biology, etc.
Related is the area of Markov chains which have   been studied extensively in various contexts\cite{M1,M2,M3}, and also the area of Markov chains Monte Carlo\cite{M4}.
A different approach (not related to the present paper) are the quantum random walks\cite{QR1,QR2} where quantum interference plays an important role. 

In this paper we study random walks in a finite Abelian group $G$ (e.g., \cite{G1,G2,G3,G4,G5}). Their study involves finite-state Markov chains with doubly stochastic transition matrices.
Most of the work on Markov chains uses row stochastic transition matrices which are more general, and
there is not much work on Markov chains with doubly stochastic transition matrices.
The merit of using doubly stochastic transition matrices is that we get stronger results (propositions \ref{pro3},\ref{pro5},\ref{pro6}). In particular:
\begin{itemize}
\item
We show that the doubly stochastic transition matrices belong to a particular subpolytope ${\cal B}(G)$ of the full Birkohff polytope ${\cal B}$, which has as vertices permutation matrices that form
a representation of the group $G$ (proposition \ref{pro3}).
\item
We study the time evolution of probability vectors describing the random walk, and prove several relations (proposition \ref{pro5}).
We also introduce polytopes that contain all future probability vectors, which do not depend on the transition matrices, and which shrink during time evolution.
Similar polytopes have been introduced in \cite{VOU1} in a different context,  but here they are polytopes of probability vectors which are related to the Birkhoff subpolytope ${\cal B}(G)$, which in turn is related to the group $G$.

\item
We consider the case of ergodic doubly stochastic transition matrices (proposition \ref{pro6}).
\end{itemize}
As examples we apply the general formalism into random walks in the additive group ${\mathbb Z}(d)$ (of integers modulo $d$), 
and also in the Heisenberg-Weyl group $HW(d)/{\mathbb Z}(d)$. We specify explicitly the polytopes of the probability vectors in these examples.

We use various quantities to describe the random walks and the associated Markov chains, as follows: 
\begin{itemize}
\item
The majorization preorder among probability vectors\cite{MAJ1,MAJ2,MAJ3}.
Related to it are the Lorenz values and the Gini index\cite{{Gini,Gini1}}, which are used to quantify the sparsity of probability vectors. 
They are statistical quantities which are very popular in Mathematical Economics, but they can be used in any context.

\item
Entropic quantities which quantify the uncertainty (opposite to sparsity) of probability vectors.

\item
The total variation distance between probability vectors \cite{G1,G2,G3,G4,G5}.

\end{itemize}

We study a physical implementation of random walks in ${\mathbb Z}(d)$ using orthogonal projectors in a sequence of non-selective quantum measurements.
 Non-selective measurements have been used recently in many applications in quantum technologies\cite{N,N1,N2,N3,N4,N5,N6,N7}.

The Heisenberg-Weyl group $HW(d)/{\mathbb Z}(d)$ is intimately related to coherent states in a quantum system with $d$-dimensional Hilbert space $H(d)$(e.g, \cite{VOU}).
We study a physical implementation of random walks on the finite Heisenberg-Weyl group $HW(d)/{\mathbb Z}(d)$, using a sequence of non-selective POVM 
(positive operator valued measure) measurements with coherent states in $H(d)$.

In section 2 we discuss two types of non-selective measurements. The first uses orthogonal projectors and 
the second uses non-selective POVM measurements with  coherent states.
In section 3 we present doubly stochastic matrices which form the Birkhoff polytope ${\cal B}$.
We then introduce Birkoff subpolytopes ${\cal B}(G)$ related to finite groups $G$.
We also introduce polytopes of probability vectors. All these polytopes play a central role in the present paper.
In section 4 we introduce various quantities that describe probability vectors.
They include, the majorization preorder among probability vectors, Lorenz values and the Gini index, and entropic quantities.

In section 5 we study random walks on a finite Abelian group $G$ using Markov chains with doubly stochastic matrices in the  Birkhoff subpolytope ${\cal B}(G)$.
Our general results are presented in propositions \ref{pro3},\ref{pro5},\ref{pro6}.
We then apply these ideas to the additive group ${\mathbb Z}(d)$, and to the Heisenberg-Weyl group $HW(d)/{\mathbb Z}(d)\simeq {\mathbb Z}(d)\times {\mathbb Z}(d)$.
We note that the applied probability literature on random walks in finite Abelian groups \cite{G1,G2,G3,G4,G5}, uses Fourier transforms to reduce convolutions (Eq.(\ref{325}) with the assumption in Eq.(\ref{36A})) into multiplications.
Here we follow another complementary approach based on Markov chains with doubly stochastic matrices and Birkhoff polytopes.

In section 6 we present a physical implementation of random walks in ${\mathbb Z}(d)$ with a sequence of non-selective projective quantum measurements.
In section 7 we present a physical implementation of random walks on the Heisenberg-Weyl group $HW(d)/{\mathbb Z}(d)$ using a sequence of non-selective POVM 
measurements with coherent states.
We conclude in section 8 with a discussion of our results.

\section{Non-selective measurements on a system with finite dimensional Hilbert space}

We  consider a quantum system $\Sigma_d$ with variables in  the ring ${\mathbb Z}(d)$ of integers modulo $d$. 
$H(d)$ is the $d$-dimensional Hilbert space describing this system. 
There are well known technical differences between quantum systems with odd dimension $d$ and even dimension $d$.
In this paper we consider systems with odd dimension $d$.

Let $|X;j\rangle$ where $j\in {\mathbb Z}(d)$ be an orthonormal basis in $H(d)$. $X$ in the notation is not a variable, it simply indicates `position states'.
The finite Fourier transform $F$ is given by
\begin{eqnarray}\label{FF}
&&F=\frac{1}{\sqrt{d}}\sum _{j,k}\omega(jk) \ket{X;j}\bra{X;k};\;\;\;\omega(\alpha)=\exp \left (i\frac{2\pi\alpha}{d}\right);\;\;\;\alpha,j,k\in{\mathbb Z}(d)\nonumber\\
&&F^4={\bf 1};\;\;\;FF^{\dagger}={\bf 1}.
\end{eqnarray}
We act with $F$ on position states and get the basis of `momentum states':
\begin{eqnarray}
&&\ket{P;j}=F\ket{X;j}.
\end{eqnarray}
$P$ in the notation is not a variable, it simply indicates `momentum states'.

The phase space of this system is  ${\mathbb Z}(d)\times {\mathbb Z}(d)$ and in it we introduce the displacement operators
\begin{eqnarray}\label{A}
&&X^\beta = \sum _j\omega (-j\beta)|P; j \rangle \bra{P; j }=\sum _j\ket{X; j+\beta}\bra{X;j};\nonumber\\
&&Z^\alpha =\sum _j|P; j+\alpha \rangle \bra{P;j}= \sum _j\omega (\alpha j)|X; j\rangle\bra{X;j}=F X^\alpha F^\dagger;\nonumber\\
&&X^{d}=Z^{d}={\bf 1};\;\;\;X^\beta Z^\alpha= Z^\alpha X^\beta \omega(-\alpha \beta);\;\;\;\alpha, \beta\in {\mathbb Z}(d).
\end{eqnarray}
General displacement operators are the unitary operators
\begin{eqnarray}
D(\alpha, \beta)=Z^\alpha X^\beta \omega (-2^{-1}\alpha \beta);\;\;\;D(\alpha, \beta)\ket{X;j}=\omega(2^{-1}\alpha\beta+\alpha j)\ket{X;j+\beta}.
\end{eqnarray}
The $2^{-1}=\frac{d+1}{2}$ is an integer in ${\mathbb Z}(d)$ with odd $d$, considered here.
The $D(\alpha, \beta)\omega (\gamma)$ with the multiplication
\begin{eqnarray}
D(\alpha_1, \beta_1)\omega(\gamma_1)D(\alpha_2, \beta_2)\omega(\gamma_2)=D(\alpha_1+\alpha_2, \beta_1+\beta_2)\omega[\gamma_1+\gamma_2+2^{-1}(\alpha_1\beta_2-\alpha_2\beta_1)],
\end{eqnarray}
form a representation of the 
Heisenberg-Weyl group $HW(d)$, which has $d^3$ elements. 

Below we use the $HW(d)/{\mathbb Z}(d)$ which has  $d^2$ elements  $D(\alpha, \beta)$ (defined up to a phase factor $\omega (\gamma)$). 
$HW(d)/{\mathbb Z}(d)$ is isomorphic to the additive group ${\mathbb Z}(d)\times {\mathbb Z}(d)$ with the componentwise addition
\begin{eqnarray}
(\alpha_1, \beta_1)+(\alpha_2, \beta_2)=(\alpha_1+\alpha_2, \beta_1+\beta_2).
\end{eqnarray}
We denote this as $HW(d)/{\mathbb Z}(d)\simeq {\mathbb Z}(d)\times {\mathbb Z}(d)$.

\subsection{Non-selective measurements with orthogonal projectors}
We consider the orthogonal projectors
\begin{eqnarray}
\varpi(j)=\ket{X;j}\bra{X;j};\;\;\;\sum_j\varpi(j)={\bf 1};\;\;\;j\in{\mathbb Z}(d).
\end{eqnarray}
We use them to perform measurements on a system with density matrix $\rho$, and we get the post-measurement density matrix $\varpi(j)$ with probability $p_{j}$:
\begin{equation}
\rho\rightarrow \varpi(j);\;\;\;p_j={\rm Tr}[\rho \varpi(j)];\;\;\;\sum _{j}p_{j}=1.
\end{equation}
If we perform non-selective measurements (in which case we do not look at the outcome \cite{N,N1,N2,N3,N4,N5,N6,N7}), the post-measurement density matrix $\sigma$ is given by
\begin{equation}\label{ns1}
\rho \rightarrow \sigma=\sum _{j}\varpi(j)\rho \varpi(j)=\sum _{j} p_{j}\varpi(j).
\end{equation}
The non-selective measurement destroys the off-diagonal elements $\ket{X;j}\bra{X;k}$ with $j\ne k$.

\subsection{Non-selective POVM measurements with coherent states}
Intimately related to $HW(d)/{\mathbb Z}(d)$ are the coherent states, which are defined as (e.g., \cite{VOU})
\begin{eqnarray}
&&\ket{C;\alpha, \beta}=D(\alpha, \beta)\ket{\eta};\;\;\;\alpha, \beta \in {\mathbb Z}_d\nonumber\\
&&\ket{\eta}=\sum _{r=0}^{d-1}\eta_r\ket{X;r};\;\;\;\sum _{r=0}^{d-1}|\eta_r|^2=1.
\end{eqnarray}
$C$ in the notation is not a variable, it simply indicates `coherent states'.
$\ket{\eta}$ is a `generic' vector (called fiducial vector), such that any $d$ coherent states are linearly independent.
Position and momentum states are examples of vectors which should not be used as fiducial vectors (in that case there are subsets of $d$ coherent states which differ only by a trivial phase factor).
We can show that
\begin{eqnarray}
\langle X;r \ket{C;\alpha, \beta}=\omega(-2^{-1}\alpha \beta +\alpha r)\eta_{r-\beta}.
\end{eqnarray}

The coherent states have the following properties:
\begin{itemize}
\item
{\bf Resolution of the identity:}
\begin{equation}
\frac{1}{d}\sum_{\alpha, \beta } \Pi(\alpha,\beta)={\bf 1};\;\;\;\Pi(\alpha,\beta)=\ket{C;\alpha, \beta}\bra{C;\alpha, \beta};\;\;\;\alpha, \beta\in {\mathbb Z}(d).
\end{equation}
$\Pi(\alpha,\beta)$ are $d^2$ projectors of rank one.
\item
{\bf Closure under displacement transformations:}
\begin{equation}
D(\gamma, \delta)\ket{C;\alpha, \beta}=\omega [2^{-1}(\beta\gamma-\alpha\delta)]\ket{C;\alpha+\gamma, \beta+\delta}.
\end{equation}
\item
{\bf Non-orthogonality:}
\begin{eqnarray}
&&\langle C;\gamma, \delta\ket{C;\alpha, \beta}=\omega[2^{-1}(\alpha\beta+\gamma\delta)-\beta\gamma]\sum _s \omega [(\alpha-\gamma)s]\eta_{s+\beta-\delta}^*\eta_s\nonumber\\
&&{\rm Tr}[\Pi(\alpha,\beta)\Pi(\gamma,\delta)]=\left |\sum _s \omega [(\alpha-\gamma)s]\eta_{s+\beta-\delta}^*\eta_s\right|^2.
\end{eqnarray}

\end{itemize}

Since 
\begin{equation}
\sum _{\alpha,  \beta}E(\alpha,\beta)={\bf 1};\;\;\;E(\alpha,\beta)=\frac{1}{d}\Pi(\alpha,\beta),
\end{equation}
the $d^2$ operators $E(\alpha,\beta)$ form a  positive operator-valued measure (POVM).
The Kraus operators $K(\alpha,\beta)$ are in this case Hermitian operators: 
\begin{equation}
K(\alpha,\beta)=\frac{1}{\sqrt{d}}\Pi(\alpha,\beta);\;\;\;E(\alpha,\beta)=[K(\alpha,\beta)]^2.
\end{equation}

We use them to perform measurements on a system with density matrix $\rho$, and we get the density matrix $\Pi(\alpha,\beta)$ with probability $p_{\alpha\beta}$:
\begin{equation}
\rho\rightarrow \Pi(\alpha,\beta);\;\;\;p_{\alpha\beta}=\frac{1}{d}{\rm Tr}[\rho \Pi(\alpha,\beta)];\;\;\;\sum _{\alpha,  \beta}p_{\alpha\beta}=1.
\end{equation}
If we perform {non-selective measurements}, the post-measurement density matrix $\sigma$ is given by
\begin{equation}\label{ns2}
\rho \rightarrow \sigma=\sum _{\alpha,  \beta} K(\alpha,\beta)\rho K(\alpha,\beta)
=\frac{1}{d}\sum _{\alpha,  \beta}\Pi(\alpha,\beta)\rho \Pi(\alpha,\beta)=\sum _{\alpha,  \beta} p_{\alpha \beta}\Pi(\alpha,\beta)
\end{equation}
The $(\alpha, \beta)\leftrightarrow d\alpha+\beta$ is a bijective map between the sets $\{0,...,d-1\}\times \{0,...,d-1\}$ and $\{0,...,d^2-1\}$.
We use it to introduce a `single index notation':
\begin{eqnarray}\label{155}
\nu=d\alpha+\beta;\;\;\;\nu=0,...,d^2-1;\;\;\;\alpha,\beta=0,...,d-1
\end{eqnarray}
Then
\begin{eqnarray}\label{155A}
\alpha=\left \lfloor{ \frac{\nu}{d}}\right \rfloor;\;\;\;\beta={\rm rem}\left ( \frac{\nu}{d}\right )
\end{eqnarray}
Here `rem' indicates the remainder.
The $D(\alpha, \beta)$, $\ket{C;\alpha, \beta}$, $\Pi(\alpha, \beta)$ will be denoted as $D(\nu)$, $\ket{C;\nu}$, $\Pi(\nu)$, correspondingly.

The above relations can now be written in a simpler way as
\begin{eqnarray}\label{14}
&&\rho \rightarrow \sigma=\sum _{\nu=0} ^{d^2-1}p_\nu\Pi(\nu);\;\;\;\Pi(\nu)=\ket{C;\nu}\bra{C;\nu}\nonumber\\
&&p_{\nu}=\frac{1}{d}{\rm Tr}[\rho \Pi(\nu)];\;\;\;;\;\;\;\sum_{\nu=0}^{d^2-1}p_\nu=1.
\end{eqnarray}
We can express $\rho$ as
\begin{eqnarray}\label{14}
&&\rho =\sigma + A;\;\;\;A=\frac{1}{d^2}\sum_{\nu\ne \mu}\ket{C;\nu}\bra{C;\mu}\bra{C;\nu}\rho\ket{C;\mu}-\left(1-\frac{1}{d}\right )\sum_\nu p_\nu\Pi(\nu);\;\;\;{\rm Tr}(A)=0.
\end{eqnarray}
The non-selective measurements destroy the off-diagonal $A$ (its trace is zero), and we are left with the diagonal part which is $\sigma$.

\subsection{A single index notation for $HW(d)/{\mathbb Z}(d)\simeq {\mathbb Z}(d)\times {\mathbb Z}(d)$}\label{sec15}

Although there is a bijective map between the {\bf sets} $\{0,...,d-1\}\times \{0,...,d-1\}$ and $\{0,...,d^2-1\}$ (given in Eqs(\ref{155}), (\ref{155A})), 
the additive group ${\mathbb Z}(d)\times {\mathbb Z}(d)$ with componentwise addition, is  {\bf not} isomorphic to the additive group ${\mathbb Z}(d^2)$ with the `standard' addition rule.
The sum in  ${\mathbb Z}(d)\times {\mathbb Z}(d)$, is {\bf not} mapped to the `standard sum' in ${\mathbb Z}(d^2)$ (a similar point has been made in ref.\cite{LV} in a different context).
For example, in ${\mathbb Z}(3)\times {\mathbb Z}(3)$
\begin{eqnarray}
(1,2)+(0,2)=(1,1).
\end{eqnarray}
But according to the map in Eq.(\ref{155}) with $d=3$,
\begin{eqnarray}
(1,2)\rightarrow 5\in {\mathbb Z}(9);\;\;\;(0,2)\rightarrow 2\in {\mathbb Z}(9);\;\;\;(1,1)\rightarrow 4\in {\mathbb Z}(9).
\end{eqnarray}
In ${\mathbb Z}(9)$ the `standard addition' gives $5+2=7$, not $4$. 
We elucidate this by writing $5$ and $2$ in the base-$3$ numerical system as 
\begin{eqnarray}
5=1\times 3^1+2\times 3^0;\;\;\;2=0\times 3^1+2\times3^0.
\end{eqnarray}
Addition in the $3^0$-column gives a `carry' which takes an extra $1$ in the $3^1$-column. There is no equivalent to `carry' in addition in ${\mathbb Z}(3)\times {\mathbb Z}(3)$.

In the single index notation, addition in ${\mathbb Z}(d)\times {\mathbb Z}(d)$ (we use the symbol $\boxplus$) or  multiplication in $HW(d)/{\mathbb Z}(d)$ gives
\begin{eqnarray}
\nu\boxplus \mu=\kappa;\;\;\;\nu, \mu, \kappa=0,...,d^2-1.
\end{eqnarray}
$\kappa$ can be calculated indirectly in two steps:
\begin{itemize}
\item
Use Eq.(\ref{155A}) to get $\nu\rightarrow (\alpha, \beta)$ and  $\mu\rightarrow (\gamma, \delta)$ where $\alpha, \beta,\gamma, \delta$ take values in
$\{0,...,d-1\}$. 
\item
Add the two pairs in ${\mathbb Z}(d)\times {\mathbb Z}(d)$, to get $(\alpha+\gamma,\beta+\delta)$ with 
$\alpha+\gamma$ and $\beta+\delta$ taking values in $\{0,...,d-1\}$, and then use Eq.(\ref{155}):
\begin{eqnarray}
(\alpha+\gamma,\beta+\delta) \rightarrow \kappa=d(\alpha+\gamma)+(\beta+\delta). 
\end{eqnarray}
\end{itemize}
With this `indirect' rule for addition the $\{0,...,d^2-1\}$ becomes an additive group,  isomorphic to the additive group ${\mathbb Z}(d)\times {\mathbb Z}(d)$ (and to the multiplicative group $HW(d)/{\mathbb Z}(d)$).

As an example we give in table \ref{t1} the sum between elements in $\{0,...,8\}$ so that it is isomorphic to the additive group ${\mathbb Z}(3)\times {\mathbb Z}(3)$.
Using this table we find $D(5)D(6)=D(5\boxplus 6)=D(2)$ in $HW(3)/{\mathbb Z}(3)$, which in the double index  notation is $D(1,2)D(2,0)=D(0,2)$ (here the displacement operators are cosets, i.e., they are defined up to a phase factor).
Also from this table we see that $7\boxplus 5=0$, and therefore the `additive inverse' of $7$ (which can also be denoted as $-7$) is $5$.
Indeed in the double index notation $7\rightarrow (2,1)$ and $5\rightarrow (1,2)$ and then $(2,1)+(1,2)=(0,0)$.
This table is used in section \ref{sec17} for multiplications of elements of $HW(3)/{\mathbb Z}(3)$ in the single index notation.

The merit of the single index notation is that later we will deal with matrices rather than tensors with four indices.

\begin{table} 
\caption{The $\{0,...,8\}$ with the addition in this table (denoted as $\boxplus$) is isomorphic to the additive group ${\mathbb Z}(3)\times {\mathbb Z}(3)$, and to the multiplicative group $HW(3)/{\mathbb Z}(3)$.}
\begin{tabular}{|c||c|c|c|c|c|c|c|c|c|}\hline
&0&1&2&3&4&5&6&7&8\\ \hhline{|=|=|=|=|=|=|=|=|=|=|}
0&0&1&2&3&4&5&6&7&8\\\hline
1&1&2&0&4&5&3&7&8&6\\\hline
2&2&0&1&5&3&4&8&6&7\\\hline
3&3&4&5&6&7&8&0&1&2\\\hline
4&4&5&3&7&8&6&1&2&0\\\hline
5&5&3&4&8&6&7&2&0&1\\\hline
6&6&7&8&0&1&2&3&4&5\\\hline
7&7&8&6&1&2&0&4&5&3\\\hline
8&8&6&7&2&0&1&5&3&4\\\hline
\end{tabular} \label{t1}
\end{table}

\section{Polytopes of doubly stochastic matrices and polytopes of probability vectors}

\subsection{Permutations and doubly stochastic matrices}

Let ${\cal I}=\{0,...,n-1\}$. A permutation $\pi$ is a bijective map $\pi(a)$ from ${\cal I}$ to itself.
The set of all permutations is the symmetric group ${\cal S}$\cite{SYM} which has $n!$ elements.
Multiplication in this group is the composition:
\begin{eqnarray}
(\pi_2 \circ \pi_1)(a)=\pi_2[\pi_1(a)];\;\;\;a=0,...,n-1.
\end{eqnarray}
The unity is ${\bf 1}(a)=a$. Each permutation $\pi$ has an inverse $\pi^{-1}$.

A permutation matrix $M_\pi$ is a $n\times n$ matrix with elements 
\begin{eqnarray}\label{PP}
M_\pi(a,b)=\delta(\pi(a),b);\;\;\;{\rm rank} (M_\pi)=n.
\end{eqnarray}
Here $\delta$ is the Kronecker delta. 
Each row and each column have exactly one element which is $1$ and the other $n-1$ elements are $0$.
It is easily seen that
\begin{eqnarray}
M_{\pi_2} M_{\pi_1}=M_{\pi_1\circ \pi_2};\;\;\;M_{\bf 1}={\bf 1};\;\;\;M_{\pi^{-1}}=[M_\pi]^{-1}=[M_\pi]^\dagger.
\end{eqnarray}
The $n!$ matrices $M_\pi$ with multiplication, form a representation of the symmetric group ${\cal S}$.

Let ${\cal P}$ be a $n\times n$ matrix with non-negative elements. It is a row stochastic matrix, if in each row the sum of all elements is equal to one:
\begin{eqnarray}\label{A7}
\sum _{b=0}^{n-1} {\cal P}_{ab}=1;\;\;{\cal P}_{ab}\ge 0.
\end{eqnarray}
It is a doubly stochastic matrix if in a addition to that, in each column the sum of all elements is equal to one:
\begin{eqnarray}\label{A7}
\sum _{a=0}^{n-1} {\cal P}_{ab}=1.
\end{eqnarray}
The product of two doubly stochastic matrices is a  doubly stochastic matrix.
The inverse of a doubly stochastic matrix might not be a doubly stochastic matrix.
The set of all doubly stochastic matrices  with respect to matrix multiplication is a semigroup.
An example of a doubly stochastic matrix is the $\frac{1}{n}J_n$, where $J_n$ is the matrix of ones (all its elements are one).

\subsection{Probability vectors}
Let ${\bf x}$ be a probability vector (written as a row)
\begin{eqnarray}
{\bf x}=(x_0,...,x_{n-1});\;\;\;x_a\ge 0;\;\;\;\sum _{a=0}^{n-1}x_a=1.
\end{eqnarray}
The product ${\bf x}{\cal P}$  is also a probability vector.
We use the convention in the area of Markov processes, which is to multiply row vectors with matrices on the right.

Special cases of probability vectors, are the most uncertain probability vector ${\bf u}$ and the most certain probability vector which is ${\bf c}$ (and its permutations):
\begin{eqnarray}\label{5}
{\bf u}=\left (\frac{1}{n},...,\frac{1}{n}\right );\;\;\;{\bf c}=(1,0,\cdots,0).
\end{eqnarray}
For any doubly stochastic matrix ${\cal P}$ we get
\begin{eqnarray}\label{6}
{\bf u}{\cal P}={\bf u}.
\end{eqnarray}
Also for any probability vector ${\bf x}$ we get
\begin{eqnarray}
{\bf x}\left (\frac{1}{n}J_n\right )={\bf u}.
\end{eqnarray}

\subsection{Birkhoff subpolytopes}

The Birkhoff-von Neumann theorem states that every doubly stochastic matrix  ${\cal P}$ can be expanded in terms of the $n!$ permutation matrices with probabilities 
$\lambda(\pi)$ as coefficients:
\begin{eqnarray}\label{20}
{\cal P}=\sum _{\pi\in {\cal S}}\lambda(\pi)M_\pi;\;\;\;\sum _{\pi \in {\cal S}}\lambda(\pi)=1;\;\;\;\lambda(\pi)\ge 0.
\end{eqnarray}
This expansion is {\bf not} unique. 

The set of all $n\times n$ doubly stochastic matrices is a convex polytope called Birkhoff polytope ${\cal B}$. It has the $n!$ permutation matrices $M_\pi$ as vertices.
The dimension of the Birkhoff polytope is $(n-1)^2$. Indeed in each column or row of the doubly stochastic matrix there are $n-1$ 
independent numbers, and the last one is determined by the constraint that the sum is one (general references on polytopes are \cite{POL1,POL2}).

We next introduce Birkhoff subpolytopes, which play an important role in the present paper. We note that every finite group is a subgroup of the symmetric group ${\cal S}$ (Cayley's theorem). 

\begin{definition}
Let ${\cal S}_1$ be a subgroup of the symmetric group ${\cal S}$ (${\cal S}_1\le {\cal S}$).
The Birkhoff subpolytope ${\cal B}(S_1)\subseteq {\cal B}$ consists of all doubly stochastic matrices which are convex combinations of the permutation matrices in ${\cal S}_1$:
\begin{eqnarray}\label{20A}
{\cal P}=\sum _{\pi\in{\cal S}_1}\lambda(\pi)M_\pi;\;\;\;\pi\in{\cal S}_1\le {\cal S};\;\;\;\sum _{\pi \in {\cal S}_1}\lambda(\pi)=1;\;\;\;\lambda(\pi)\ge 0.
\end{eqnarray}
\end{definition}
Products of doubly stochastic matrices in ${\cal B}(S_1)$ belong in ${\cal B}(S_1)$.
The matrix ${\bf 1}$ belongs in ${\cal B}(S_1)$.
If a permutation matrix $M_\pi\in {\cal B}(S_1)$ then its inverse $M_{\pi^{-1}}=M_\pi^\dagger \in {\cal B}(S_1)$.

\begin{example}\label{ex1}
We consider the additive group ${\mathbb Z}(n)$, and the  $n\times n$
permutation matrix
\begin{eqnarray}\label{30A}
X=\begin{pmatrix}
0&0&0&\cdots&1\\
1&0&0&\cdots&0\\
0&1&0&\cdots&0\\
\vdots&\vdots&\vdots&\ddots&\vdots\\
0&0&0&\cdots&0\\
\end{pmatrix};\;\;\;X^n={\bf 1}.
\end{eqnarray}
The operator $X$ in Eq.(\ref{A}) in the position basis, is the above matrix.
The permutation matrices ${\bf 1}, X, \cdots, X^{n-1}$ form a representation of ${\mathbb Z}(n)$, which is a subgroup of the symmetric group ${\cal S}$.
The Birkhoff subpolytope ${\cal B}[{\mathbb Z}(n)]$ contains all doubly stochastic matrices ${\cal P}$ which are convex combinations of the $n$ permutation matrices ${\bf 1}, X, \cdots, X^{n-1}$:
\begin{eqnarray}
 {\cal P}=\sum _{a=0}^{n-1}p(a)X^a;\;\;\;\sum _{a=0}^{n-1}p(a)=1.
 \end{eqnarray}
We note that $\frac{1}{n}J$ belongs in ${\cal B}[{\mathbb Z}(n)]$ and 
\begin{eqnarray}
 \frac{1}{n}J_n=\frac{1}{n}\sum _{a=0}^{n-1}X^a.
 \end{eqnarray}

\end{example}

\begin{example}\label{ex2}
We consider the additive group ${\mathbb Z}(n)\times {\mathbb Z}(n)\simeq HW(n)/{\mathbb Z}(n)$.
The $n^2$ permutation matrices 
\begin{eqnarray}\label{100}
M_{na+b}=X^a\otimes X^b;\;\;a,b=0,...,n-1, 
\end{eqnarray}
form a representation of the group ${\mathbb Z}(n)\times {\mathbb Z}(n)$, which is a subgroup of the symmetric group ${\cal S}$. We use the single index notation for labelling the permutation matrices $M$.

For $n=3$ we give some examples of the $9\times 9$ permutation matrices $M_{na+b}$ (written in a block notation):
\begin{eqnarray}
M_0={\bf 1};\;\;\;M_4=X\otimes X=\begin{pmatrix}
{\bf 0}&{\bf 0}&X\\
X&{\bf 0}&{\bf 0}\\
{\bf 0}&X&{\bf 0}
\end{pmatrix};\;\;\;
M_7=X^2 \otimes X=\begin{pmatrix}
{\bf 0}&X&{\bf 0}\\
{\bf 0}&{\bf 0}&X\\
X&{\bf 0}&{\bf 0}
\end{pmatrix}.
\end{eqnarray}

The Birkhoff subpolytope ${\cal B}[{\mathbb Z}(n)\times {\mathbb Z}(n)]$ contains all doubly stochastic matrices ${\cal P}$ which are convex combinations of the $n^2$ permutation matrices $M_{na+b}$.
\begin{eqnarray}
 {\cal P}=\sum _{r=0}^{n^2-1}p(r)M_r;\;\;\;\sum _{r=0}^{n^2-1}p(r)=1.
\end{eqnarray}
We note that $\frac{1}{n^2}J_{n^2}$ belongs in ${\cal B}[{\mathbb Z}(n)\times {\mathbb Z}(n)]$ and 
\begin{eqnarray}
 \frac{1}{n^2}J_{n^2}=\frac{1}{n^2}\sum _{\nu=0}^{n^2-1}M_\nu.
 \end{eqnarray}

\end{example}

\subsection{Polytopes of probability vectors}
If ${\bf x}$ is a probability vector then ${\bf y}={\bf x}{\cal P}$ is also a probability vector, and using Eq.(\ref{20}) we get
\begin{eqnarray}\label{21}
{\bf y}={\bf x}{\cal P}=\sum _{\pi}\lambda(\pi){\bf x}M_\pi;\;\;\;\sum _{\pi}\lambda(\pi)=1.
\end{eqnarray}
${\bf y}$ is a convex combination of vectors with components which are permutations of the components of  ${\bf x}$.
\begin{definition}\label{def1}
We represent each probability vector ${\bf x}=(x_0,...,x_{n-1})$ by a point in $(x_0,...,x_{n-2})\in [0,1]^{n-1}$ (which involves the first $n-1$ components, because the last component $x_{n-1}$ is not independent).
Then:
\begin{itemize}
\item
The polytope of a  probability vector ${\bf x}$ contains all probability vectors  ${\bf y}={\bf x}{\cal P}$,
 for all doubly stochastic matrices ${\cal P}\in{\cal B}$:
\begin{eqnarray}\label{112}
{\cal A}({\bf x})=\left\{{\bf x}{\cal P} |{\cal P}\in{\cal B}\right \}\subset[0,1]^{n-1}.
\end{eqnarray}
\item
The polytope of a  probability vector ${\bf x}$ with respect to the Birkhoff subpolytope ${\cal B}({\cal S}_1)$, contains all probability vectors  ${\bf y}={\bf x}{\cal P}$
 for all doubly stochastic matrices ${\cal P}\in{\cal B}({\cal S}_1)$:
\begin{eqnarray}\label{112A}
{\cal A}[{\bf x};{\cal B}({\cal S}_1)]=\left\{{\bf x}{\cal P} |{\cal P}\in{\cal B}({\cal S}_1)\right \}\subseteq {\cal A}({\bf x})\subset[0,1]^{n-1};\;\;\;{\cal S}_1\le {\cal S}.
\end{eqnarray}
\end{itemize}

\end{definition}
${\cal A}({\bf x})$ is a convex polytope that has ${\cal V}$  probability vectors ${\bf x}M_\pi$ as vertices (where ${\cal V}\le n!$ because some of them might be equal to each other).
The most uncertain probability vector ${\bf u}$ (Eq.(\ref{5})) belongs to ${\cal A}({\bf x})$.

We next define the simplex $\Delta _{n-1}$ which contains {\bf all} ($n$-dimensional) probability vectors ${\bf x}$:
\begin{eqnarray}
\Delta_{n-1}=\left \{(x_0,...,x_{n-2})\; |\;\sum _{i=0}^{n-2}x_i\le 1\right \}\subset[0,1]^{n-1}.
\end{eqnarray}
$\Delta_{n-1}$ is $(n-1)$-dimensional, and ${\cal A}({\bf x})\subseteq \Delta_{n-1}$. 

In the case of the most uncertain probability vector ${\bf u}$ in Eq.(\ref{5}), ${\cal A}({\bf u})=\{{\bf u}\}$ consists of a single point which can be viewed as the zero-dimensional simplex $\Delta _0$. Therefore
\begin{eqnarray}
\Delta_0\subseteq {\cal A}({\bf x})\subseteq \Delta_{n-1}.
\end{eqnarray}
The dimension of ${\cal A}({\bf x})$ is
\begin{eqnarray}
0\le {\rm dim}[{\cal A}({\bf x})]\le n-1.
\end{eqnarray}
From the definition \ref{def1} follows immediately that 
\begin{eqnarray}\label{12}
{\cal A}({\bf x}{\cal P})\subseteq {\cal A}({\bf x});\;\;\;{\cal P}\in{\cal B}.
\end{eqnarray}
Also
\begin{eqnarray}\label{12A}
{\cal A}[{\bf x}{\cal P};{\cal B}({\cal S}_1)]\subseteq {\cal A}[{\bf x};{\cal B}({\cal S}_1)];\;\;\;{\cal P}\in{\cal B}({\cal S}_1).
\end{eqnarray}
\begin{example}
We consider the probability vector ${\bf x}=(0.5,0.5,0)$. In this case there are $3$ different from each other vectors $M_\pi{\bf x}$.
\begin{eqnarray}
{\bf x}=(0.5,0.5,0);\;\;\;{\bf x}_1=(0,0.5,0.5);\;\;\;{\bf x}_2=(0.5,0,0.5).
\end{eqnarray}
Then ${\cal A}({\bf x})\subset [0,1]^2$ is a triangle with vertices $(0.5,0)$ and $(0, 0.5)$ and $(0.5,0.5)$.
\end{example}
\begin{example}
We consider the probability vector ${\bf x}=(0.5,0.5,0,0)$. In this case there are $6$ different from each other vectors $M_\pi{\bf x}$.
\begin{eqnarray}
{\bf x}=(0.5,0.5,0,0);\;\;\;{\bf x}_1=(0.5,0,0.5,0);\;\;\;{\bf x}_2=(0.5,0,0,0.5)\nonumber\\
{\bf x}_3=(0,0.5,0.5,0);\;\;\;{\bf x}_4=(0,0.5,0,0.5);\;\;\;{\bf x}_5=(0,0,0.5,0.5).
\end{eqnarray}
Then ${\cal A}({\bf x})\subset [0,1]^3$ is a polytope with $6$ vertices 
\begin{eqnarray}
(0.5,0.5,0);\;\;\;(0.5,0,0.5);\;\;\;(0.5,0,0);\;\;\;(0,0.5,0.5);\;\;\;(0,0.5,0);\;\;\;(0,0,0.5).
\end{eqnarray}

\end{example}

\subsection{Ergodic doubly stochastic matrices}

A $n\times n$ doubly stochastic matrix ${\cal P}$ has the eigenvalue $e_0=1$, and the other eigenvalues $e_1,...,e_{n-1}$ have absolute value $|e_i|\le 1$.
The left (row) eigenvector corresponding to the eigenvalue $e_0=1$, is ${\bf u}$ (Eq.(\ref{6})).
Therefore any doubly stochastic matrix can be written as
\begin{eqnarray}\label{743}
{\cal P}={\bf u}^\dagger {\bf u}+\sum _{r=1}^{n-1}e_r{\bf v}_r^\dagger{\bf v}_r=\frac{1}{n}J_n+\sum _{r=1}^{n-1}e_r{\bf v}_r^\dagger{\bf v}_r.
\end{eqnarray}
Here ${\bf v}_r$ is the left (row) eigenvector corresponding to the eigenvalue $e_r$.
The doubly stochastic matrix $\frac{1}{n}J_n$  has the eigenvalues $e_0=1$ and $e_r=0$ for $r=1,...,n-1$. 

\begin{proposition}\label{pro1}
Let ${\cal P}$ be a doubly stochastic matrix such that all eigenvalues of ${\cal P}$ apart from the eigenvalue $e_0=1$,
have absolute value  $|e_i|<1$. Then
\begin{eqnarray}\label{AA1}
\lim _{k\rightarrow\infty}{\cal P}^k=\frac{1}{n}J_n.
\end{eqnarray}
In this case we say that the doubly stochastic matrix is ergodic.
\end{proposition}
\begin{proof}
${\cal P}^k$ has the eigenvalues $e_0^k=1$, and $e_r^k$ with $|e_r|^k\rightarrow\;0$ as $k\rightarrow\infty$.
Then 
\begin{eqnarray}
{\cal P}^k={\bf u}^\dagger{\bf u}+\sum _{r=1}^{n-1}(e_r)^k{\bf v}_r^\dagger{\bf v}_r\rightarrow \frac{1}{n}J_n.
\end{eqnarray}

\end{proof}
The $e_{\rm max}=\max \{|e_1|,...,|e_{n-1}|\}$ defines how quickly we will reach the limit in Eq.(\ref{AA1}).
The $(e_{\rm max})^k$ should be close to zero
\begin{eqnarray}
e_{\rm max}^k\approx 0\Rightarrow k\gg \frac {1}{|\log e_{\rm max}|}.
\end{eqnarray}

\section{Quantities that describe probability vectors}

\subsection{Sparsity of probability vectors: Majorization, Lorenz values and the Gini index}

Lorenz values and the Gini index\cite{Gini,Gini1} are statistical quantities which are mainly used in Mathematical Economics (for their use in a quantum context see ref\cite{VOU2}).
Related is the preorder majorization \cite{MAJ1,MAJ2,MAJ3}.
Here we introduce briefly these quantities for probability vectors, in order to establish the notation.
We also state without proof, some of their known properties which are needed later.

 Let $\pi$ be the permutation that orders the probabilities of a probability vector ${\bf x}$, in ascending order: 
\begin{eqnarray}\label{pp}
x[\pi(0)]\le x[\pi(1)]\le ...\le x[\pi(n-1)].
\end{eqnarray}

The Lorenz values of this probability vector are defined as
\begin{eqnarray}
{\cal L}(a;{\bf x})=x[\pi(0)]+...+ x[\pi(a)];\;\;\;{\cal L}(n-1;{\bf x})=1;\;\;\;a=0,...,n-1.
\end{eqnarray}
The Lorenz values are cumulative probabilities with respect to the ascending order.
For a general probability vector
\begin{eqnarray}\label{24}
0\le {\cal L}(a;{\bf x})\le \frac{a+1}{n};\;\;\;{\cal L}(a;{\bf x})={\cal L}(a;{\bf x}M_\pi).
\end{eqnarray}
The Lorenz values for the most uncertain probability vector ${\bf u}$ and the certain probability vector ${\bf c}$, are
\begin{eqnarray}
{\cal L}(a;{\bf u})=\frac{a+1}{n};\;\;\;{\cal L}(a;{\bf c})=\delta(a, n-1).
\end{eqnarray}

Majorization \cite{MAJ1,MAJ2,MAJ3} is a preorder based on Lorenz values and is defined here for probability vectors as follows.
 The probability vector  ${\bf x}$ majorizes ${\bf y}$ (or ${\bf x}$ is more sparse than ${\bf y}$) if ${\cal L}(a;{\bf x})\le {\cal L}(a;{\bf y})$ for all $a$.
This is denoted as ${\bf x}\succ {\bf y}$.
For any probability vector ${\bf x}$
\begin{eqnarray}\label{19}
{\bf u}\prec{\bf x}\prec {\bf c}.
\end{eqnarray}
It is known \cite{MAJ1,MAJ2,MAJ3} that ${\bf x}\succ {\bf y}$ if and only if there exists a doubly stochastic matrix ${\cal P}$ such that ${\bf y}={\cal P}{\bf x}$:
\begin{eqnarray}\label{AA}
{\bf x}\succ {\bf y}\Leftrightarrow {\bf y}={\cal P}{\bf x}.
\end{eqnarray}

The Gini index is defined by the following relations which are known to be equivalent to each other:
\begin{eqnarray}\label{85}
{\cal G}({\bf x})&=&1-\frac{2}{n+1}\sum _{a=0}^{n-1}{\cal L}(a;{\bf x})=
1-\frac{2}{n+1}\{nx[\pi(0)]+(n-1)x[\pi(1)]+...+x[\pi(n-1)]\}\nonumber\\&=&\frac{1}{2(n+1)}\sum _{a,b}|x_a-x_b|
\end{eqnarray}
It takes values in the interval
\begin{eqnarray}\label{www}
0\le {\cal G}({\bf x})\le \frac{n-1}{n+1}.
\end{eqnarray}
Large values of the Gini index indicate sparse (certain) probability vectors.
Small values of the Gini index indicate uncertain probability vectors.
For the vectors ${\bf u}$ and ${\bf c}$ we get:
\begin{eqnarray}
{\cal G}({\bf u})=0;\;\;\;{\cal G}({\bf c})=\frac{n-1}{n+1}
\end{eqnarray}

If ${\cal P}$ is a doubly stochastic matrix, then 
\begin{eqnarray}\label{PA1}
{\cal G}({\bf x}{\cal P})\le {\cal G}({\bf x}).
\end{eqnarray}
For a permutation matrix $M_\pi$ we get
\begin{eqnarray}\label{P2}
{\cal G}({\bf x}M_{\pi})={\cal G}({\bf x}).
\end{eqnarray}

\subsection{Total variation distance between two probability vectors}

Let ${\bf x}$ and ${\bf y}$ be two probability vectors. ${\cal I}=\{0,...,n-1\}$ is the set of their indices.
If $A\subseteq {\cal I}$ we define 
\begin{eqnarray}
x(A)=\sum_{a\in A}x_a.
\end{eqnarray}

The total variation distance between these two probability vectors is defined as\cite{G1,G2,G3,G4,G5}
\begin{eqnarray}\label{346}
||{\bf x}-{\bf y}||=\sup_{A\subseteq {\cal I}}|x(A)-y(A)|=\frac{1}{2}\sum _a|x_a-y_a|;\;\;\;0\le ||{\bf x}-{\bf y}||\le 1.
\end{eqnarray}
In relation to this we prove the following:
\begin{itemize}
\item
The two expressions in Eq.(\ref{346}) are indeed equal to each other. The supremum is achieved for the set $A_1=\{a|x_a\ge y_a)\}$ and then
\begin{eqnarray}
||{\bf x}-{\bf y}||=|x(A_1)-y(A_1)|=\sum _{a\in A_1}(x_a-y_a);\;\;x_a>y_a.
\end{eqnarray}
It is easily seen that the right hand side is equal to $\frac{1}{2}\sum _a|x_a-y_a|$.
\item
Since $0\le x(A)\le 1$ and $0\le y(A)\le 1$, we conclude that
\begin{eqnarray}
0\le ||{\bf x}-{\bf y}||\le 1.
\end{eqnarray}
\item
$||{\bf x}-{\bf y}||$ is a distance because it is easily seen that
\begin{eqnarray}
&&||{\bf x}-{\bf x}||=0;\;\;\;||{\bf x}-{\bf y}||\ge 0;\;\;\;||{\bf x}-{\bf y}||=||{\bf y}-{\bf x}||\nonumber\\
&&||{\bf x}-{\bf y}||+||{\bf y}-{\bf z}||\ge ||{\bf x}-{\bf z}||.
\end{eqnarray}
For later use we prove the following lemma
\begin{lemma}
For any doubly stochastic matrix ${\cal P}$
\begin{eqnarray}\label{45}
||{\bf x}-{\bf y}||\ge ||{\bf x}{\cal P}-{\bf y}{\cal P}||.
\end{eqnarray}
If ${\bf y}={\bf u}$, then
\begin{eqnarray}\label{46A}
||{\bf x}-{\bf u}||\ge ||{\bf x}{\cal P}-{\bf u}||.
\end{eqnarray}
\end{lemma}
\begin{proof}
\begin{eqnarray}
||{\bf x}{\cal P}-{\bf y}{\cal P}||&=&\frac{1}{2}\sum_b\left | \sum_a x_a{\cal P}_{ab}-y_a{\cal P}_{ab}\right |\le \frac{1}{2}\sum_b\sum_a\left | x_a{\cal P}_{ab}-y_a{\cal P}_{ab}\right |\nonumber\\
&=&\frac{1}{2}\sum_a\left | x_a-y_a\right |\left (\sum _b{\cal P}_{ab}\right )=\frac{1}{2}\sum_a\left | x_a-y_a\right |=||{\bf x}-{\bf y}||.
\end{eqnarray}
If ${\bf y}={\bf u}$, then ${\bf u}{\cal P}={\bf u}$ and follows Eq.(\ref{46A}).

\end{proof}

\end{itemize}

\subsection{Entropic quantities}

We define the entropy of a probability vector as
\begin{eqnarray}\label{86}
E({\bf x})=-\sum _ax_a\log x_a;\;\;\;0\le E({\bf x})\le \log n
\end{eqnarray} 
$E({\bf c})=0$ for the most certain probability vector, and $E({\bf u})=\log n$ for the most uncertain vector.

It is known that for any doubly stochastic matrix ${\cal P}$ 
\begin{eqnarray}\label{11}
E({\bf x})\le E({\bf x}{\cal P}).
\end{eqnarray}

\section{Random walks on finite Abelian groups using Markov chains with doubly stochastic matrices in a Birkhoff subpolytope}

We consider random walks on a finite Abelian multiplicative group $G$. The approach in refs. \cite{G1,G2,G3,G4,G5} was to use Fourier transforms 
to reduce convolutions that appear in the formalism (Eq.(\ref{325}) below with the assumption in Eq.(\ref{36A})) into multiplications. 
Problems discussed in these references are: how quickly we get close to  stationarity, the cutoff phenomenon\cite{G3}, shuffling of cards, etc.

Here we follow another complementary approach based on 
Markov chains with doubly stochastic matrices in the Birkhoff subpolytope ${\cal B}(G)$.
After the general formalism we apply these ideas to two examples: the group ${\mathbb Z}(d)$ and the group $HW(d)/{\mathbb Z}(d)\simeq  {\mathbb Z}(d)\times {\mathbb Z}(d)$.

We assume that $G$ has $\ell$ elements $g_0,...,g_{\ell-1}$. $1$ is the unit element, and $g^{-1}$ is the inverse of $g$.  To every element $g_a\in G$, we associate a probability $p(g_a)$ such that
\begin{eqnarray}\label{10A}
\sum _{a=0}^{\ell-1}p(g_a)=1.
\end{eqnarray}
We consider a random walk in the group $G$, with the time taking discrete values $t=0,1,...$. At time $t=0$ we have  the element $g_{a_0}$ with probability $p(g_{a_0})$
(each index has another index which shows the discrete time).  At time $t=n$ we get the element $g_{a_n}$ with probability 
\begin{eqnarray}\label{325}
q^{(n)}(g_{a_n})=\sum _{a_0,...,a_{n-1}}p(g_{a_0})p(g_{a_0},g_{a_1})\cdots p(g_{a_{n-1}},g_{a_n});\;\;\;a_r=0,...,\ell-1.
\end{eqnarray}
The superfix $(n)$ indicates the discrete time.
Here we assume the Markov property, according to which the next step depends only on the present step, and is independent of the past steps. For example, from $t=0$ to $t=1$ the transition probability is
$p(g_{a_0},g_{a_1})$ and does not depend on what happened the previous times $t=-1, -2,...$. Also
\begin{eqnarray}\label{3}
\sum _{a_{k+1}=0}^{\ell-1} p(g_{a_k},g_{a_{k+1}})=1.
\end{eqnarray}
This is because from $g_{a_k}$, we will go to some element from the group with probability one.
We can rewrite Eq.(\ref{325}) in terms of matrices as
\begin{eqnarray}\label{321}
q^{(n)}=q^{(0)}{\cal P}(0,1)\cdots {\cal P}(n-1,n).
\end{eqnarray}
Here:
\begin{itemize}
\item
 ${\cal P}(k,k+1)$  is an $\ell \times \ell$ transition matrix with elements ${\cal P}_{ab}(k,k+1)=p(g_{a},g_{b})$.
All transition matrices ${\cal P}(0,1)$, ${\cal P}(1,2)$, etc, are assumed to be equal to each other and will be denoted as ${\cal P}$.
The transition matrix ${\cal P}$ is a row stochastic matrix, because of Eq.(\ref{3}).
\item
$q^{(n)}(g_{a_n})$ (with $a_n=0,...,\ell-1$) is the probability vector that gives the probability to be in the element $g_{a_n}$ at time $n$. 
Also $q^{(0)}(g_{a_0})=p(g_{a})$.
\end{itemize}
Then Eq.(\ref{321}) becomes
\begin{eqnarray}\label{320}
q^{(n)}=q^{(0)}{\cal P}^n.
\end{eqnarray}

As in refs\cite{G1,G2,G3,G4,G5} we make the following extra assumption.
The transition probability is taken to be 
\begin{eqnarray}\label{36A}
p(g_a,g_b)=p(g_a^{-1}g_b).
\end{eqnarray}
$p(g_a,g_b)$ depends only on  $g_a^{-1}g_b$.
In this case
\begin{eqnarray}\label{30}
\sum _{a=0}^{\ell-1} p(g_a,g_b)=\sum _{a=0}^{\ell-1}p(g_a^{-1}g_b)=1.
\end{eqnarray}
With this extra assumption, the transition matrix ${\cal P}$ is the $\ell\times\ell$ doubly stochastic matrix
\begin{eqnarray}\label{111}
{\cal P}=\begin{pmatrix}
p(1)&p(g_0^{-1}g_1)&p(g_0^{-1}g_2)&\cdots&p(g_0^{-1}g_{\ell-1})\\
p(g_1^{-1}g_0)&p(1)&p(g_1^{-1}g_2)&\cdots &p(g_1^{-1}g_{\ell-1})\\
\vdots&\vdots&\vdots&\vdots&\vdots\\ 
p(g_{\ell-1}^{-1}g_0)&p(g_{\ell-1}^{-1}g_1)&p(g_{\ell-1}^{-1}g_2)&\cdots&p(1)\\
\end{pmatrix}
\end{eqnarray}
The probabilities $p(g_a)$ in Eq.(\ref{10A}) define uniquely the doubly stochastic transition matrix ${\cal P}$.

\begin{proposition}\label{pro3}
\mbox{}
\begin{itemize}
\item[(1)]
If one of the probabilities $p(g_r)=1$ (and $p(g_s)=0$ for $r\ne s$), then the matrix  in Eq.(\ref{111}) is a permutation matrix denoted as $M(r)$.
\item[(2)]
The $\ell$ permutation matrices $M(r)$ form a subgroup of the symmetric group , which is isomorphic to $G$.
\item[(3)]
The matrices ${\cal P}$ in Eq.(\ref{111}) form the Birkoff subpolytope ${\cal B}(G)$.
\end{itemize}
\end{proposition}
\begin{proof}
\mbox{}
\begin{itemize}
\item[(1)]
We show that when $p(g_r)=1$ (and $p(g_s)=0$ for $r\ne s$), each row and each column of $M(r)$ have exactly one element which is $1$, and the other $n-1$ elements are $0$.

Each row and each column have all probabilities $p(g_0)$,...,$p(g_{\ell-1})$ (each appearing once).
The $(a+1)$-row has the element $p(g_a^{-1}g_b)$ in the $(b+1)$-column, and is $1$ in the $(b+1)$-column such that $g_a^{-1}g_b=g_r$ (therefore $g_b=g_ag_r$)
and $0$ in all other columns.
Therefore each row has exactly one element equal to $1$, and all the others are $0$.
In a similar way we prove that  each column  has exactly one element equal to $1$, and all the others are $0$.
This proves that when $p(g_r)=1$, then $M(r)$ is a permutation matrix.

\item[(2)]
If ${\cal I}=\{0,..., \ell-1\}$ is the set of indices, the relation $g_a^{-1}g_b=g_r$ creates a bijective map $b=\pi_r(a)$ from ${\cal I}$ to itself, i.e., a permutation. 
We call ${\cal S}_1$ the set of these $\ell$ permutations. To each element $g_r\in G$ corresponds a permutation $\pi_r\in {\cal S}_1$.
We will show that ${\cal S}_1$ with multiplication (composition) is a subgroup of the symmetric group , which is isomorphic to $G$.

For all $a\in {\cal I}$, we get
\begin{eqnarray}
g_a^{-1}g_{\pi_r(a)}=g_r.
\end{eqnarray}
Also if $a=\pi_s(c)$ then
\begin{eqnarray}
g_c^{-1}g_a=g_c^{-1}g_{\pi_s(c)}=g_s.
\end{eqnarray}
We multiply these two relations, and we get
\begin{eqnarray}
g_c^{-1}g_{\pi_r(a)}=g_c^{-1}g_{\pi_r[\pi_s(c)]}=g_c^{-1}g_{\pi_r\circ \pi_s(c)}=g_rg_s.
\end{eqnarray}
This proves that to the element $g_rg_s\in G$ corresponds the permutation $\pi_r\circ \pi_s$.
So the product (composition) of permutations in ${\cal S}_1$ belongs in ${\cal S}_1$ (closure property).
To the unity element $1\in G$ corresponds the unity permutation ${\bf 1}\in {\cal S}_1$.
Also to $g_r^{-1}\in G$ corresponds the permutation $\pi_r^{-1}\in {\cal S}_1$.
Therefore ${\cal S}_1$ is a group isomorphic to $G$.
In terms of permutation matrices, to $\pi_r$ corresponds the $M_{ab}(r)=\delta(\pi_r(a),b)$. 

\item[(3)]
${\cal P}$ is a linear function of the various probabilities. 
By definition $M(r)$ is a permutation matrix, that we get when $p(g_r)=1$ (and $p(g_s)=0$ for $r\ne s$) in ${\cal P}$. It follows that
\begin{eqnarray}
{\cal P}=\sum_{r=0}^{\ell-1}p(g_r)M(r).
\end{eqnarray}
Therefore the matrices ${\cal P}$ in Eq.(\ref{111}) are convex combinations of the permutation matrices $M(r)$ which form a representation of the group $G$.
Consequently the matrices ${\cal P}$ form the Birkoff subpolytope ${\cal B}(G)$.

\end{itemize}
\end{proof}

\begin{proposition}\label{pro5}
We consider a random walk on a finite Abelian group $G$ with the doubly stochastic matrix ${\cal P}$ in Eq.(\ref{111}). Then
\begin{itemize}
\item[(1)]
All probability vectors $q^{(n)}$ with $n\ge k$ are in the subpolytope ${\cal A}[q^{(k)};{\cal B}(G)]$ (defined in general in Eq.(\ref{112A})) which depends only on $q^{(k)}$ and {\bf does not depend} on the transition matrix ${\cal P}$.
\item[(2)]
During time evolution:
\begin{itemize}
\item
The probability vectors at different times are ordered according to the majorization preorder as follows:
\begin{eqnarray}
q^{(0)}\succ q^{(1)}\succ q^{(2)}\succ \cdots.
\end{eqnarray}
\item
The entropy of the probability vectors increases:
\begin{eqnarray}
E[q^{(0)})]\le E[q^{(1)}]\le E[q^{(2)}]\le \cdots.
\end{eqnarray}
\item
The Gini index of the probability vectors decreases:
\begin{eqnarray}
{\cal G}[q^{(0)})]\ge {\cal G}[q^{(1)}]\ge {\cal G}[q^{(2)}]\ge \cdots.
\end{eqnarray}

\item
The polytopes of the probability vectors shrink in time:
\begin{eqnarray}
{\cal A}[q^{(0)};{\cal B}(G)]\supseteq {\cal A}[q^{(1)};{\cal B}(G)]\supseteq {\cal A}[q^{(2)};{\cal B}(G)]\supseteq \cdots .
\end{eqnarray}

\end{itemize}

\end{itemize}
\end{proposition}
\begin{proof}
\mbox{}
\begin{itemize}
\item[(1)]
By definition, the polytope ${\cal A}[q^{(k)};{\cal B}(G)]$ contains all vectors $q^{(k)}{\cal P}^{n-k}$ for any doubly stochastic transition matrix in the Birkhoff subpolytope ${\cal B}(G)$.
Therefore it contains all probability vectors $q^{(n)}$ with $n\ge k$.
\item[(2)]
If $m\le n$ we get $q^{(n)}=q^{(m)}{\cal P}^{n-m}$, and from Eqs (\ref{AA}), (\ref{11}), (\ref{PA1}), (\ref{12A}) follows  that 
\begin{eqnarray}
q^{(m)}\succ q^{(n)};\;\;\;E(q^{(m)})\le E(q^{(n)});\;\;\;{\cal G}[q^{(m)})]\ge {\cal G}[q^{(n)}];\;\;\;{\cal A}[q^{(m)};{\cal B}(G)]\supseteq {\cal A}[q^{(n)};{\cal B}(G)]
\end{eqnarray}

\end{itemize}
\end{proof}

\begin{proposition}\label{pro6}
We consider a random walk on a finite Abelian group $G$ with the {\bf ergodic} doubly stochastic matrix ${\cal P}$ in Eq.(\ref{AA1}). Then
\begin{eqnarray}
\lim _{k\rightarrow\infty}q^{(k)}={\bf u};\;\;\;\lim _{k\rightarrow\infty}E[q^{(k)})]=\log \ell;\;\;\;\lim _{k\rightarrow\infty}{\cal G}[q^{(k)})]=0.
\end{eqnarray}
Also the polytope ${\cal A}(q^{(k)})$ will shrink into the simplex $\Delta_0$ (a single point), in the  $k\rightarrow\infty$ limit.
In this case during time evolution the total variation distance of $q^{(k)}$ from ${\bf u}$ decreases:
\begin{eqnarray}\label{45A}
||q^{(0)}-{\bf u}||\ge ||q^{(1)}-{\bf u}||\ge ||q^{(2)}-{\bf u}||\ge \cdots.
\end{eqnarray}
\end{proposition}
\begin{proof}
Using Eq.(\ref{AA1}) we find that 
\begin{eqnarray}
\lim _{k\rightarrow\infty}q^{(k)}= q^{(0)}\lim _{k\rightarrow\infty}{\cal P}^k=q^{(0)}\left (\frac{1}{\ell}J_\ell \right )={\bf u}.
\end{eqnarray}
Therefore
\begin{eqnarray}
\lim _{k\rightarrow\infty}E[q^{(k)})]=E({\bf u})=\log \ell;\;\;\;\lim _{k\rightarrow\infty}{\cal G}[q^{(k)})]={\cal G}({\bf u})=0.
\end{eqnarray}
We prove Eq.(\ref{45A}) using Eq.(\ref{46A}).

\end{proof}

\begin{definition}
When proposition \ref{pro6} holds, we define the mixing time as
 \begin{eqnarray}
{\mathfrak N}(\epsilon)=\min _n\{||q^{(n)}-{\bf u}||\le \epsilon\};\;\;\;\epsilon>0.
\end{eqnarray}
\end{definition}
There is much work on upper bounds for the mixing times in various cases\cite{G1,G2,G3,G4,G5}, which is not discussed here.

\subsection{Random walks in the additive group ${\mathbb Z}(d)$}\label{sec20}
In this section we perform random walks in the additive group ${\mathbb Z}(d)$ (integers modulo $d$).
Let $p(a)$ be the probability for the element $a$ with
\begin{eqnarray}
\sum p(a)=1;\;\;\;a=0,...,d-1.
\end{eqnarray}
For simplicity we consider ${\mathbb Z}(5)$ and we find that 
${\cal P}$ is the doubly stochastic matrix
\begin{eqnarray}\label{89B}
{\cal P}=\begin{pmatrix}
p(0)&p(1)&p(2)&p(3)&p(4)\\
p(4)&p(0)&p(1)&p(2)&p(3)\\
p(3)&p(4)&p(0)&p(1)&p(2)\\
p(2)&p(3)&p(4)&p(0)&p(1)\\
p(1)&p(2)&p(3)&p(4)&p(0)\\
\end{pmatrix}=p(0){\bf 1}+p(4)X+p(3)X^2+p(2)X^3+p(1)X^4.
\end{eqnarray}
Here $X$ is the permutation matrix in Eq.(\ref{30A}) for $d=5$.
The ${\bf 1}, X, X^2, X^3,X^4$ are permutation matrices, and form a representation of the  group ${\mathbb Z}(5)$, which  is a subgroup of the symmetric group ${\cal S}$.
Therefore ${\cal P}$ belongs to the Birkhoff subpolytope ${\cal B}[{\mathbb Z}(5)]$.

\begin{example}
We assume that in the matrix ${\cal P}$ in Eq.(\ref{89B}) we have
\begin{eqnarray}
p(0)=p(1)=0.5;\;\;\;p(2)=p(3)=p(4)=0.
\end{eqnarray}
Then ${\cal P}$ has the eigenvalue $e_0=1$, and the maximum absolute value of the other four eigenvalues is $e_{\rm max}=0.80$.

We start the random walk at $t=0$ from the element $0\in {\mathbb Z}(5)$, and using Eq.(\ref{86}) for the entropy and Eq.(\ref{85}) for the Gini index we find
\begin{eqnarray}\label{600B}
q^{(0)}=\begin{pmatrix}
1&0&0&0&0
\end{pmatrix};\;\;\;
E(q^{(0)})=0;\;\;\;{\cal G}(q^{(0)})=0.5
\end{eqnarray}
The polytope ${\cal A}(q^{(0)})\subset [0,1]^4$ has $5$ vertices 
\begin{eqnarray}\label{601B}
(0,0,0,0);\;\;\;(1,0,0,0);\cdots;(0,0,0,1).
\end{eqnarray}

At $t=8$ using Eq.(\ref{86}) for the entropy and Eq.(\ref{85}) for the Gini index, we get
\begin{eqnarray}
&&q^{(8)}=q^{(0)}{\cal P}^8=\begin{pmatrix}
0.222&0.140&0.140&0.222&0.276
\end{pmatrix}\nonumber\\
&&E(q^{(8)})=1.574({\rm nats});\;\;\;{\cal G}(q^{(8)})=0.118
\end{eqnarray}
So at $t=8$ we are in the element $0$ with probability $0.222$, in the element $1$ with probability $0.140$ etc.

Since $e_{\rm max}<1$ it follows that ${\cal P}$ is ergodic doubly stochastic matrix, and as $n\rightarrow\infty$ we get $q^{(n)}\rightarrow {\bf u}$.
We calculated the total variation distances
\begin{eqnarray}
||q^{(0)}-{\bf u}||=0.80;\;\;\;||q^{(8)}-{\bf u}||=0.12.
\end{eqnarray}

\end{example}

\subsection{Random walks in $HW(d)/{\mathbb Z}(d)\simeq  {\mathbb Z}(d)\times {\mathbb Z}(d)$}\label{sec17}
In this section we perform random walks in $HW(d)/{\mathbb Z}(d)\simeq  {\mathbb Z}(d)\times {\mathbb Z}(d)$.
In this case $\ell=d^2$, and we use the following `single index notation' introduced earlier:
\begin{eqnarray}
&&\nu=d\alpha+\beta;\;\;\;g_\nu=D(\alpha, \beta);\;\;\;p(\nu)=p(g_\nu)=p[D(\alpha, \beta)]\nonumber\\
&&\sum _\nu p(\nu)=1; \;\;\;\nu=0,\cdots,d^2-1.
\end{eqnarray}

For simplicity we consider random walk in  $HW(3)/{\mathbb Z}(3)$. Taking into account section \ref{sec15} about the multiplication $D(\nu)D(\mu)$ (and in particular table \ref{t1}), we find that
${\cal P}$ is the doubly stochastic matrix
\begin{eqnarray}\label{89C}
{\cal P}=\begin{pmatrix}
p(0)&p(1)&p(2)&p(3)&p(4)&p(5)&p(6)&p(7)&p(8)\\
p(2)&p(0)&p(1)&p(5)&p(3)&p(4)&p(8)&p(6)&p(7)\\
p(1)&p(2)&p(0)&p(4)&p(5)&p(3)&p(7)&p(8)&p(6)\\
p(6)&p(7)&p(8)&p(0)&p(1)&p(2)&p(3)&p(4)&p(5)\\
p(8)&p(6)&p(7)&p(2)&p(0)&p(1)&p(5)&p(3)&p(4)\\
p(7)&p(8)&p(6)&p(1)&p(2)&p(0)&p(4)&p(5)&p(3)\\
p(3)&p(4)&p(5)&p(6)&p(7)&p(8)&p(0)&p(1)&p(2)\\
p(5)&p(3)&p(4)&p(8)&p(6)&p(7)&p(2)&p(0)&p(1)\\
p(4)&p(5)&p(3)&p(7)&p(8)&p(6)&p(1)&p(2)&p(0)\\
\end{pmatrix}.
\end{eqnarray}
Using the matrix $X$ in Eq.(\ref{30A}) with $n=3$ we rewrite this in a block notation as
\begin{eqnarray}
{\cal P}=\begin{pmatrix}
p(0){\bf 1}+p(1)X^2+p(2)X&p(3){\bf 1}+p(4)X^2+p(5)X&p(6){\bf 1}+p(7)X^2+p(8)X\\
p(6){\bf 1}+p(7)X^2+p(8)X&p(0){\bf 1}+p(1)X^2+p(2)X&p(3){\bf 1}+p(4)X^2+p(5)X\\
p(3){\bf 1}+p(4)X^2+p(5)X&p(6){\bf 1}+p(7)X^2+p(8)X&p(0){\bf 1}+p(1)X^2+p(2)X\\
\end{pmatrix}.
\end{eqnarray}
From this follows that
\begin{eqnarray}
{\cal P}&=&p(0){\bf 1}\otimes {\bf 1}+p(1){\bf 1}\otimes X^2+p(2){\bf 1}\otimes X+p(3)X^2\otimes {\bf1}+p(4)X^2\otimes X^2\nonumber\\&+&
p(5)X^2\otimes X+p(6)X\otimes {\bf 1}+p(7)X\otimes X^2+p(8)X\otimes X.
\end{eqnarray}
Using the notation in Eq.(\ref{100}) we rewrite this as
\begin{eqnarray}\label{89E}
{\cal P}&=&p(0)M_0+p(1)M_2+p(2)M_1+p(3)M_6+p(4)M_8\nonumber\\&+&
p(5)M_7+p(6)M_3+p(7)M_5+p(8)M_4.
\end{eqnarray}
The $M_0,...,M_8$ are permutation matrices, and form a representation of the  $HW(3)/{\mathbb Z}(3)\simeq  {\mathbb Z}(3)\times {\mathbb Z}(3)$ group, which  is a subgroup of the symmetric group ${\cal S}$.
Therefore ${\cal P}$ belongs to the Birkhoff subpolytope ${\cal B}[HW(3)/{\mathbb Z}(3)]$.

\begin{example}
As an example, we assume that in the matrix ${\cal P}$ in Eq.(\ref{89E}) we have
\begin{eqnarray}
p(0)=p(1)=0.3;\;\;\;p(2)=p(3)=p(4)=p(5)=p(6)=0;\;\;\;p(7)=p(8)=0.2.
\end{eqnarray}
Then ${\cal P}$ has the eigenvalue $e_0=1$, and the maximum absolute value of the other eight eigenvalues is $e_{\rm max}=0.53$.

We start the random walk at $t=0$ from the element ${\bf 1}\in HW(3)/{\mathbb Z}(3)$ in which case using Eq.(\ref{86}) for the entropy and Eq.(\ref{85}) for the Gini index, we get
\begin{eqnarray}\label{600}
q^{(0)}=\begin{pmatrix}
1&0&0&0&0&0&0&0&0
\end{pmatrix};\;\;\;
E(q^{(0)})=0;\;\;\;{\cal G}(q^{(0)})=0.80
\end{eqnarray}
The polytope ${\cal A}(q^{(0)})\subset [0,1]^8$ has $9$ vertices 
\begin{eqnarray}\label{601}
(0,0,0,0,0,0,0,0);\;\;\;(1,0,0,0,0,0,0,0);\cdots;(0,0,0,0,0,0,0,1).
\end{eqnarray}

At $t=4$ using Eq.(\ref{86}) for the entropy and Eq.(\ref{85}) for the Gini index, we get
\begin{eqnarray}
&&q^{(4)}=q^{(0)}{\cal P}^4=\begin{pmatrix}
0.088&0.088&0.106&0.108&0.129&0.108&0.139&0.116&0.116
\end{pmatrix}\nonumber\\
&&E(q^{(4)})=2.184({\rm nats});\;\;\;{\cal G}(q^{(4)})=0.073.
\end{eqnarray}
So at $t=4$ we are in the element $D(0)=D(0,0)={\bf 1}$ with probability $0.088$, in the element $D(1)=D(0,1)$ with probability $0.088$, in the element $D(2)=D(0,2)$ with probability $0.106$, etc.

Since $e_{\rm max}<1$ it follows that ${\cal P}$ is ergodic doubly stochastic matrix, and as $n\rightarrow\infty$ we get $q^{(n)}\rightarrow {\bf u}$.
We calculated the total variation distances
\begin{eqnarray}
||q^{(0)}-{\bf u}||=0.888;\;\;\;||q^{(4)}-{\bf u}||=0.056.
\end{eqnarray}
\end{example}

\begin{example}
As another example, we assume that in the matrix ${\cal P}$ in Eq.(\ref{89E}), we have the binomial distribution of probabilities
\begin{eqnarray}
p(\nu)=\begin{pmatrix}
8\\
\nu
\end{pmatrix}
f^\nu(1-f)^{8-\nu};\;\;\;f=0.3;\;\;\;\nu=0,...,8.
\end{eqnarray}
Then ${\cal P}$ has the eigenvalue $e_0=1$, and  the maximum absolute value of the other eight eigenvalues is $e_{\rm max}=0.493$.

We start the random walk at $t=0$ from the element ${\bf 1}$ in which case $q^{(0)}$, $E(q^{(0)})$, ${\cal G}(q^{(0)})$,
and the polytope ${\cal A}(q^{(0)})$ are the same as in the previous example (given in Eqs(\ref{600}),(\ref{601})).

At $t=3$ using Eq.(\ref{86}) for the entropy and Eq.(\ref{85}) for the Gini index, we get
\begin{eqnarray}
&&q^{(3)}=q^{(0)}{\cal P}^3=\begin{pmatrix}
0.087&   0.087&    0.092&    0.139&    0.146&    0.119&    0.105&    0.099&    0.121
\end{pmatrix}\nonumber\\
&&E(q^{(3)})=2.174({\rm nats});\;\;\;{\cal G}(q^{(3)})=0.099
\end{eqnarray}

${\cal P}$ is ergodic doubly stochastic matrix (because $e_{\rm max}<1$), and as $n\rightarrow\infty$ we get $q^{(n)}\rightarrow {\bf u}$.
We calculated the total variation distances
\begin{eqnarray}
||q^{(0)}-{\bf u}||=0.888;\;\;\;||q^{(3)}-{\bf u}||=0.135.
\end{eqnarray}

\end{example}

\section{Implementation of the random walks in ${\mathbb Z}(d)$ with a sequence of non-selective projective measurements}\label{sec44}

In this section we present a physical implementation of the random walks in ${\mathbb Z}(d)$, with a sequence of non-selective projective measurements.
We assume that at $t=0$ the system is described with a $d\times d$ density matrix $\rho(0)$.
We perform the non-selective measurement in Eq.(\ref{ns1}) and we get the following post-measurement density matrix at $t=0$:
\begin{eqnarray}
\sigma(0)=\sum _j q_j^{(0)}\varpi(j);\;\;\;q_j^{(0)}={\rm Tr}[\rho(0) \varpi(j)];\;\;\;j=0,...,d-1.
\end{eqnarray}
$q_j^{(0)}$ is a probability vector.
Between $t=0$ and  $t=1$ the density matrix $\sigma(0)$ evolves with the unitary matrix $U$:
\begin{eqnarray}\label{123}
\rho(1)=\sum _j q_j^{(0)}U\varpi(j)U^\dagger;\;\;\;j=0,...,d-1.
\end{eqnarray}
We then perform again the non-selective measurement in Eq.(\ref{ns1}) and we get the following post-measurement density matrix at $t=1$:
\begin{eqnarray}
&&\sigma (1)=\sum _{j}q_j^{(1)}\varpi(j);\;\;\;q_j^{(1)}={\rm Tr}[\rho (1)\varpi(j)]=\sum _k q_k^{(0)}V_{kj}\nonumber\\
&&V_{kj}={\rm Tr}\left [U\varpi (k)U^\dagger\varpi(j)\right ]=|\bra{X;j}U\ket{X;k}|^2.
\end{eqnarray}
It is easily seen that $V$ is a $d\times d$ doubly stochastic matrix. 
But in order to implement a random walk in ${\mathbb Z}(d)$, it is {\bf essential} that the matrix $V$ belongs to the Birkhoff subpolytope ${\cal B}[{\mathbb Z}(d)]$
(proposition \ref{pro3}).
A simple example is to take $U=X$ (or a power of $X$).
In this case $V=X$.
 We note here, that if all the matrix elements in $U$ are $0$ or $1$ (as in the case of permutation matrices), then $V=U$.

The density matrix $\sigma(1)$ evolves between $t=1$ and $t=2$ into the density matrix 
\begin{eqnarray}
\rho(2)=U\sigma(1)U^\dagger.
\end{eqnarray}
The unitary matrix $U$ describing the time evolution between $t=1$ and $t=2$, is taken to be the same as the unitary matrix describing the time evolution between $t=0$ and $t=1$.
We then perform again the non-selective  measurement in Eq.(\ref{ns1}) and we get the post-measurement density matrix at $t=2$:
\begin{eqnarray}
&&\sigma (2)=\sum _{j}q_j^{(2)}\varpi(j);\;\;\;q_j^{(2)}={\rm Tr}[\rho (2)\varpi(j)]=\sum _k q_k^{(1)}V_{kj}=\sum _{r} q_r^{(0)}(V^2)_{rj} 
\end{eqnarray}
 We continue in a similar way, and get the sequence
\begin{eqnarray}\label{700}
\rho(0)\xrightarrow[\rm projective \; meas]{\rm non-selective}\sigma (0)\xrightarrow{U}\rho(1)\xrightarrow[\rm projective \; meas]{\rm non-selective}\sigma (1)\xrightarrow{U}\rho(2)\xrightarrow
[\rm projective \; meas]{\rm non-selective}\sigma (2)\xrightarrow{U}\cdots
\end{eqnarray}
It is a sequence of non-selective measurements (as in Eq.(\ref{ns1})) at discrete times, with unitary transformations that describe time evolution between them.
The relationship between the various probability vectors after the measurements is
\begin{eqnarray}\label{801}
q^{(n)}=q^{(k)}V^{n-k};\;\;\;n>k;\;\;\;V\in {\cal B}[{\mathbb Z}(d)];\;\;\;q^{(n)}\in {\cal A}[q^{(k)};{\cal B}({\mathbb Z}(d))].
\end{eqnarray}
$V$ is a doubly stochastic matrix which belongs to the Birkhoff subpolytope ${\cal B}[{\mathbb Z}(d)]$.
$q^{(n)}$ is a probability vector that belongs to the polytope ${\cal A}[q^{(k)};{\cal B}({\mathbb Z}(d))]$.

At time $t=n$ the post-measurement density matrix is 
\begin{eqnarray}
\sigma(n)=\sum _j q_j^{(n)}\varpi(j);\;\;\;q^{(n)}=q^{(0)}V^n.
\end{eqnarray}

Above we assumed that $V=U=X$ or any of the vertices $X^r$ in the Birkhoff subpolytope ${\cal B}[{\mathbb Z}(d)]$.
General matrices within this polytope are not unitary, and their use requires a more general process where the unitary evolution between measurements 
is replaced with a random unitary channel\cite{R1,R2,R3} (open system), which are completely positive and trace preserving maps. 
In this case we have a set of unitary operators $U^{(a)}$ with probabilities ${\mathfrak p}_a$, and time-evolution of a density matrix $\sigma$ is described as 
\begin{eqnarray}\label{WW}
\sigma\rightarrow \sum _a{\mathfrak p}_aU^{(a)}\sigma [U^{(a)}]^\dagger;\;\;\;\sum _a{\mathfrak p}_a=1.
\end{eqnarray}
We take $U^{(a)}=X^a$ and Eq.(\ref{123}) is generalised into 
\begin{eqnarray}
\rho(1)=\sum _{k=0} ^{d-1}\left (q_k^{(0)}\sum_{a=0}^{d-1} {\mathfrak p}_a X^a\varpi(k)X^{-a}\right ).
\end{eqnarray}
The post-measurement density matrix at $t=1$, becomes 
\begin{eqnarray}\label{ABC}
&&\sigma (1)=\sum _{j}q_j^{(1)}\varpi(j);\;\;\;q_j^{(1)}={\rm Tr}[\rho (1)\varpi(j)]=\sum _k q_k^{(0)}V_{kj}\nonumber\\
&&V_{kj}=\sum_{a=0}^{d-1} {\mathfrak p}_a{\rm Tr}\left [X^a\varpi (k)X^{-a}\varpi(j)\right ]=\sum_{a=0}^{d-1} {\mathfrak p}_a|\bra{X;j}X^a\ket{X;k}|^2=\sum_{a=0}^{d-1}{\mathfrak p}_a(X^a)_{kj}.
\end{eqnarray}
The matrix $V$ is now in the interior of the Birkhoff subpolytope ${\cal B}[{\mathbb Z}(d)]$.
When we replace the unitary evolution with a random unitary channel (for an open system), everything that we discussed in this section holds with $V$  the matrix in Eq.(\ref{ABC}).

\section{Implementation of the random walks in $HW(d)/{\mathbb Z}(d)$ with a sequence of non-selective POVM measurements}\label{sec45}

In this section we follow the same process as in the previous section, but here we perform the non-selective POVM measurements in Eq.(\ref{ns2}).
We intentionally use a  similar notation as in the previous section, so that the analogy is clear.
Of course here the probability vectors and the transition matrices are $d^2$-dimensional, in contrast to the previous section where they were $d$-dimensional.

At $t=0$ the system is described with the density matrix $\rho(0)$ and we perform the  non-selective POVM measurement in Eq.(\ref{ns2}). We get the post-measurement density matrix
\begin{eqnarray}
\sigma(0)=\sum _\nu q_\nu^{(0)}\Pi(\nu);\;\;\;q_\nu^{(0)}=\frac{1}{d}{\rm Tr}[\rho(0) \Pi(\nu)];\;\;\;\nu=0,...,d^2-1.
\end{eqnarray}
Between $t=0$ and $t=1$ the system evolves as a random unitary channel (open system) and is described with Eq.(\ref{WW}).
The density matrix $\sigma(0)$ becomes
\begin{eqnarray}\label{123A}
\rho(1)=\sum _{\nu=0}^{d^2-1}\left (q_\nu^{(0)} \sum_a{\mathfrak p}_aU^{(a)}\Pi(\nu)[U^{(a)}]^\dagger\right ).
\end{eqnarray}
We then perform again the non-selective POVM measurement in Eq.(\ref{ns2}) and we get the following post-measurement density matrix at $t=1$:
\begin{eqnarray}
&&\sigma (1)=\sum _{\mu}q_\mu^{(1)}\Pi(\mu);\;\;\;q_\mu^{(1)}=\frac{1}{d}{\rm Tr}[\rho (1)\Pi(\mu)]=\sum _\nu q_\nu^{(0)}W_{\nu \mu} \nonumber\\
&&W_{\nu\mu}=\frac{1}{d}\sum _a{\mathfrak p}_a{\rm Tr}\left [U^{(a)}\Pi(\nu)[U^{(a)}]^\dagger\Pi(\mu)\right ]=\frac{1}{d}\sum _a{\mathfrak p}_a|\bra{C;\mu}U^{(a)}\ket{C;\nu}|^2.
\end{eqnarray}
$W_{\nu\mu}$ is a $d^2\times d^2$ doubly stochastic matrix (the proof is based on the resolution of the identity of coherent states).
But in order to implement a random walk in $HW(d)/{\mathbb Z}(d)$, it is {\bf essential} that the matrix $W$ belongs to the Birkhoff subpolytope ${\cal B}[HW(d)/{\mathbb Z}(d)]$.
We take $U^{(a)}=D(a)$ (the displacement operators in the one index notation). 
Then 
\begin{eqnarray}
&&W_{\nu\mu}=\frac{1}{d}\sum _{a=0}^{d^2-1}{\mathfrak p}_a|\bra{C;\mu}D(a)\ket{C;\nu}|^2=\frac{1}{d}\sum _{a=0}^{d^2-1}{\mathfrak p}_a|\bra{C;\mu}{C;\alpha \boxplus \nu}\rangle |^2.
\end{eqnarray}
If we take ${\mathfrak p}_a=\frac{1}{n^2}$, then for all values $of a$ the $\alpha \boxplus \nu$ takes all values, and taking into account the resolution of the identity of coherent states we get
\begin{eqnarray}
&&W_{\nu\mu}=\frac{1}{d^3}\sum _{a=0}^{d^2-1}|\bra{C;\mu}{C;\alpha \boxplus \nu}\rangle |^2=\frac{1}{d^2}.
\end{eqnarray}
In this case $W=\frac{1}{n^2}J_{n^2}$ which belongs in the Birkhoff subpolytope ${\cal B}[HW(d)/{\mathbb Z}(d)]$, as discussed in example \ref{ex2}.
For small enough perturbations ${\mathfrak p}_a=\frac{1}{n^2}+\epsilon_a$ (with $\sum \epsilon_a=0$) we will get a matrix $W$ close to $W=\frac{1}{n^2}J_{n^2}$ which will be within
the Birkhoff subpolytope ${\cal B}[HW(d)/{\mathbb Z}(d)]$.

The density matrix $\sigma(1)$ evolves between $t=1$ and $t=2$ as in Eq.(\ref{123A}) into 
\begin{eqnarray}
\rho(2)=\sum _{\nu=0}^{d^2-1}\left (q_\nu^{(1)} \sum_a{\mathfrak p}_aU^{(a)}\Pi(\nu)[U^{(a)}]^\dagger\right ).
\end{eqnarray}
We then perform again the non-selective POVM measurement in Eq.(\ref{ns2}) and we get the post-measurement density matrix at $t=2$:
\begin{eqnarray}
&&\sigma (2)=\sum _{\mu}q_\mu^{(2)}\Pi(\mu);\;\;\;q_\mu^{(2)}=\frac{1}{d}{\rm Tr}[\rho (2)\Pi(\mu)]=\sum _\mu q_\nu^{(1)}W_{\nu \mu} =
\sum _{\kappa} q_\kappa^{(0)}(W^2)_{\kappa\mu}
\end{eqnarray}
We continue in a similar way:
\begin{eqnarray}\label{701}
\rho(0)\xrightarrow[\rm POVM\; meas]{\rm non-selective}\sigma (0)\xrightarrow[\rm unitary\; channel]{\rm random}\rho(1)\xrightarrow[\rm POVM\; meas]{\rm non-selective}\sigma (1)
\xrightarrow[\rm unitary\; channel]{\rm random}\rho(2)\rightarrow
\end{eqnarray}
The relationship between the various probability vectors after the measurements is
\begin{eqnarray}\label{802}
q^{(n)}=q^{(k)}W^{n-k};\;\;\;n>k;\;\;\;W\in {\cal B}[HW(d)/{\mathbb Z}(d)];\;\;\;q^{(n)}\in {\cal A}\{q^{(k)};{\cal B}[HW(d)/{\mathbb Z}(d)]\}
\end{eqnarray}
$W$ is a doubly stochastic matrix which belongs to the Birkhoff subpolytope ${\cal B}[HW(d)/{\mathbb Z}(d)]$.
$q^{(n)}$ is a probability vector that belongs to the polytope ${\cal A}\{q^{(k)};{\cal B}[HW(d)/{\mathbb Z}(d)]\}$.

At time $t=n$ the post-measurement density matrix is 
\begin{eqnarray}
\sigma(n)=\sum _\nu q_\nu^{(n)}\Pi(\nu);\;\;\;q^{(n)}=q^{(0)}W^n.
\end{eqnarray}

\section{Discussion}

Random walks are models for the physical process of diffusion and have been studied extensively in various contexts.
In this paper we study random walks in a finite Abelian group $G$. Their study involves Markov chains with doubly stochastic transition matrices in the
Birkhoff subpolytope ${\cal B}(G)$ (proposition \ref{pro3}).
We have shown that for each probability vector $q^{(k)}$ at time $k$, there exist a polytope ${\cal A}[q^{(k)};{\cal B}(G)]$ which contains all future probability vectors, and which does not depend on the transition matrix.
These polytopes shrink in time(proposition \ref{pro5}).
We have also studied ergodic doubly stochastic transition matrices (proposition \ref{pro6}).

The general theory has been applied to random walks in the additive group ${\mathbb Z}(d)$ (section \ref{sec20}), 
and in the Heisenberg-Weyl group $HW(d)/{\mathbb Z}(d)\simeq {\mathbb Z}(d)\times {\mathbb Z}(d)$ (section \ref{sec17}). 
The doubly stochastic matrices for these cases are given explicitly in Eqs.(\ref{89B}), (\ref{89E}).
We also presented numerical examples for these two cases, and characterised the probability vectors with quantities like entropy, the Gini index, and the total variation distance from the most uncertain 
probability vector ${\bf u}$. 

In section \ref{sec44}, we presented a physical implementation of random walks in ${\mathbb Z}(d)$ using the sequence of non-selective quantum measurements in Eq.(\ref{700}) with orthogonal projectors.
The probability vectors in this case are given in Eq.(\ref{801}).

In section \ref{sec45}, we presented a physical implementation of random walks in the  Heisenberg-Weyl group $HW(d)/{\mathbb Z}(d)\simeq {\mathbb Z}(d)\times {\mathbb Z}(d)$, 
using the sequence of non-selective POVM measurements in Eq.(\ref{701}) with coherent states.
The probability vectors in this case are given in Eq.(\ref{802}).

The work is a contribution to random walks on finite Abelian groups, using Markov chains with doubly stochastic matrices. We emphasised the importance of Birkhoff subpolytopes and also polytopes of probability vectors, for this study. We also used a variety of quantities (entropy, the Gini index, the total variation distance, etc) to characterise the probability vectors describing the random walks.

\end{document}